\documentclass[a4paper,12pt]{article}
\usepackage{amssymb,amsmath,amsthm}
\usepackage{texdraw}
\usepackage{a4wide}
\usepackage{pgf,tikz}
\usetikzlibrary{arrows}

\usepackage{subcaption}

\newtheorem{theorem}{Theorem}[section]

\newtheorem{proposition}[theorem]{Proposition}

\theoremstyle{definition}

\newtheorem{remark}[theorem]{Remark}

\begin{document}


\title{Relativistic Collisions as Yang--Baxter maps}
\author{Theodoros E. Kouloukas 
\\
\\School of Mathematics, Statistics \& Actuarial Science, \\ University of Kent, UK}
\maketitle

\begin{abstract}
We prove that one-dimensional elastic relativistic collisions satisfy the set-theoretical Yang--Baxter equation. The corresponding collision 
maps are symplectic and admit a Lax representation. Furthermore, they can be considered as reductions  
of a higher dimensional integrable Yang--Baxter map on an invariant manifold.  In this framework, we study 
the integrability of transfer maps that represent particular periodic sequences of collisions.

\end{abstract}

\section{Introduction to Yang--Baxter maps}
Solutions of the set-theoretical Yang--Baxter equation \cite{Baxter,buch,Drin,skly88,ves2,wein,Yang}, as well as their applications have been extensively studied from various  
perspectives of mathematics and physics. We will use the short term ``Yang--Baxter maps''  for these solutions which was introduced in \cite{ves2}. Namely, for a set $\mathcal{X}$ 
a  {\it{Yang--Baxter (YB) map}} is a map $R: \mathcal{X} \times
\mathcal{X} \rightarrow \mathcal{X} \times \mathcal{X}$,
$R:(x,y)\mapsto (u(x,y),v(x,y))$, that satisfies the {\em
YB equation} 
\begin{equation*}\label{YBprop}
 R_{23}\circ R_{13}\circ
R_{12}=R_{12}\circ R_{13}\circ R_{23}.
\end{equation*} Here, by $R_{ij}$
for $i,j=1,2,3$, we denote the action of the map $R$ on the $i$ and
$j$ factor of $\mathcal{X} \times \mathcal{X} \times \mathcal{X}$,
i.e. $R_{12}(x,y,z)=(u(x,y),v(x,y),z)$,
$R_{13}(x,y,z)=(u(x,z),y,v(x,z))$ and
$R_{23}(x,y,z)=(x,u(y,z),v(y,z))$. 
Moreover, $R$ is called quadrirational if both maps 
$u( \cdot ,y)$, $v(x, \cdot )$, for fixed $y$ and $x$ respectively, are birational isomorphisms  
of $\mathcal{X}$ to itself.

Parametric versions of YB maps are closely related to integrable quadrilateral equations which constitute discrete analogue of partial differential equations (see e.g. \cite{sokor,kp4,paptonlift,paptong,paptonves}). 
In these cases the YB property reflects the 3-dimensional (3D) consistency property of the corresponding equation.
Parametric YB maps \cite{ves2,ves3} involve two parameters in a parameter space $\mathcal{I}$ that can be considered as extra variables 
invariant under the map. Particularly, a \emph{parametric YB map} is a YB map  
 $R:(\mathcal{X} \times
\mathcal{I}) \times (\mathcal{X} \times \mathcal{I}) \mapsto
(\mathcal{X} \times \mathcal{I}) \times (\mathcal{X} \times
\mathcal{I})$, with
\begin{equation} \label{pYB}
R:((x,\alpha),(y,\beta))\mapsto((u,\alpha),(v,\beta))= ((u(x,\alpha,
y,\beta),\alpha),(v(x,\alpha, y,\beta),\beta)). 
\end{equation}
We usually
keep the parameters separate and denote the map (\ref{pYB}) just by 
$R_{\alpha,\beta}(x,y):\mathcal{X}\times \mathcal{X} \rightarrow
\mathcal{X} \times \mathcal{X}$.
A classification of parametric YB maps on $\mathbb{CP}^1 \times \mathbb{CP}^1$ has been achieved 
in \cite{ABS2,papclas}.
According to \cite{ves4}, a {\em Lax matrix} of the parametric YB
map (\ref{pYB}) is a matrix $L$ that depends on a point $x \in \mathcal{X}$,  
a parameter $\alpha \in \mathcal{I}$ and a spectral parameter $\zeta\in \mathbb{C}$, such
that
\begin{equation} \label{laxmat}
L(u,\alpha,\zeta)L(v,\beta,\zeta)=L(y,\beta,\zeta)L(x,\alpha,\zeta).
\end{equation}
Moreover, we call $L$ a {\em
strong Lax matrix}
if equation (\ref{laxmat}) is equivalent to 
$(u,v)=R_{\alpha,\beta}(x,y)$. 
On the other hand, 
if $u=u_{\alpha,\beta}(x,y), v=v_{\alpha,\beta}(x,y) $ satisfy
(\ref{laxmat}) for a matrix $L$  and the equation $$L(
\hat{x},\alpha,\zeta )L( \hat{y} ,\beta,\zeta )L(\hat{z}, \gamma,\zeta )= L(x
,\alpha,\zeta)L(y, \beta,\zeta)L(z, \gamma,\zeta)$$ implies that $\hat{x}=x, \
\hat{y}=y$ and $\hat{z}=z$ for every $x,y,z \in \mathcal{X}$, then
$R_{\alpha,\beta}(x,y)\mapsto(u,v)$ is a parametric YB map with Lax
matrix $L$ \cite{kp1,ves2}.

The integrability aspects of YB maps have been studied by Veselov
\cite{ves2,ves3}. It was shown that for any YB map there is a hierarchy of commuting
transfer maps. From the corresponding Lax representation of the original YB map a monodromy matrix is defined, whose spectrum is preserved under 
the transfer maps. In many cases of multidimensional YB maps with polynomial Lax matrices, r-matrix Poisson structures provide the right framework  to study the integrability of the transfer maps \cite{kp1,kp3}.  

\section{One-dimensional elastic collisions as YB maps} 
The one dimensional non-relativistic elastic collision of two particles with masses $m_1$ and $m_2$ is described 
by 
$$v_1'=\frac{v_1(m_1-m_2)+2 m_2 v_2}{m_1+m_2}, \  v_2'=\frac{v_2(m_2-m_1)+2 m_1 v_1}{m_1+m_2} ,$$
where $v_1$, $v_2$ are the initial velocities and $v_1'$, $v_2'$  the velocities after the collision. 
A simple calculation shows that the linear map 
$R^0_{m_1,m_2}(v_1,v_2)=(v_1',v_2')$ is a parametric YB map. In this case, the corresponding transfer maps are linear and their dynamical behavior quite simple. 

The relativistic case is much more interesting.  In this case, the conservation of relativistic energy and momentum is expressed as 
\begin{eqnarray}
m_1 c^2 \gamma(v_1)+m_2 c^2 \gamma(v_2) &=& m_1 c^2 \gamma(v_1')+m_2 c^2 \gamma(v_2') , \label{conen} \\ 
m_1 v_1 \gamma(v_1)+m_2 v_2 \gamma(v_2) &=& m_1 v_1' \gamma(v_1')+m_2 v_2 \gamma(v_2'),  \label{conmom}
\end{eqnarray}
where $ \gamma$ denotes the Lorentz factor $\gamma(v)=\frac{1}{\sqrt{1-v^2 / c^2}}$ (we will always assume $|v_i|<c$). The velocities $v_1', v_2'$ after collision are given in terms of the initial velocities $v_1,v_2$ by solving \eqref{conen}-\eqref{conmom}. From this solution (excluding the trivial 
solution  $v_1'=v_1, v_2'=v_2$) we define the {\it collision map}  
\begin{equation} \label{YB1}
\mathcal{R}_{m_1,m_2}:(v_1,v_2) \mapsto (v_1',v_2'). 
\end{equation} 
We are going to show that this map is a parametric YB map and we will describe an associated Lax representation. 

First, by setting 
\begin{equation} \label{transf}
x=\sqrt{\frac{c+v_1}{c-v_1}} ,  \ y=\sqrt{\frac{c+v_2}{c-v_2}} , \ u=\sqrt{\frac{c+v_1'}{c-v_1'}}  , \ v=\sqrt{\frac{c+v_2'}{c-v_2'}},  
\end{equation}
the equations \eqref{conen}-\eqref{conmom} become 
\begin{eqnarray}
m_1 \frac{x^2+1}{x}+m_2 \frac{y^2+1}{y} &=& m_1 \frac{u^2+1}{u}+m_2 \frac{v^2+1}{v}, \label{conen2} \\ 
m_1 \frac{x^2-1}{x}+m_2 \frac{y^2-1}{y} &=& m_1 \frac{u^2-1}{u}+m_2 \frac{v^2-1}{v}. \label{conmom2}
\end{eqnarray}
The latter system gives two solutions with respect to $u$ and $v$: the trivial solution $u=x$, $v=y$, which corresponds to the no-collision situation, and 
the after-collision solution 
\begin{equation} \label{uv}
u=y \left( \frac{m_1 x+m_2 y}{m_2 x+m_1 y} \right), \ v=x \left( \frac{m_1 x+m_2 y}{m_2 x+m_1 y} \right).
\end{equation}
 From the non-trivial solution \eqref{uv} we define the map 
${R}_{m_1,m_2}:(x,y) \mapsto (u,v)$.

\begin{proposition} \label{prop1}
The map 
\begin{equation} \label{YBab}
{R}_{\alpha,\beta}(x,y)=\left(y \left( \frac{\alpha x+\beta y}{\beta x+\alpha y}\right),x \left(\frac{\alpha x+\beta y}{\beta x+\alpha y} \right) \right)  
\end{equation} 
is a parametric YB map with strong Lax matrix 
\begin{equation} \label{Lax}
L(x,\alpha,\zeta)=\left(
\begin{array}{cc}
  \zeta & \alpha x \\
 \frac{\alpha}{x} &  \zeta
\end{array}
\right).
\end{equation}
Furthermore,  $R_{\alpha,\beta}$ is symplectic with respect to 
$\omega=\left( \frac{1}{x^2} -\frac{1}{y^2} \right) dx \wedge dy.$
\end{proposition}
\begin{proof}
The equation $L(u,\alpha,\zeta)L(v,\beta,\zeta)=L(y,\beta,\zeta)L(x,\alpha,\zeta)$ admits the unique solution 
$u=y \frac{\alpha x+\beta y}{\beta x+\alpha y}$, $v= x \frac{\alpha x+\beta y}{\beta x+\alpha y}$. In addition, 
$L(
\hat{x},\alpha )L( \hat{y} ,\beta )L(\hat{z}, \gamma )= L(x
,\alpha)L(y, \beta)L(z, \gamma)$ implies that $\hat{x}=x, \
\hat{y}=y$ and $\hat{z}=z$.  Finally, we can show directly that 
$R_{\alpha,\beta}^* \omega=\omega$.

\end{proof}

The parametric YB map \eqref{YBab} is reversible, i.e. $R^{21}_{\beta,\alpha}\circ R_{\alpha,\beta}=Id$ and  
it is an involution, 
$R_{\alpha,\beta}\circ R_{\alpha,\beta}=Id$. 
Furthermore, it is a quadrirational YB map of subclass $[2:2]$ which  
can be classified according to \cite{ABS2,papclas} as it is shown in remark  \ref{remclas}. 
The invariant symplectic form $\omega$ is a special case of an invariant $2n$-volume form of $2n$-dimensional transfer maps 
which is presented in Proposition \ref{propInt}.

The transformation \eqref{transf} indicates that  besides the real positive parameters $\alpha=m_1$ and $\beta=m_2$ (which correspond to the masses of the particles) we have to consider the variables $x$ and $y$ of the map $R_{\alpha,\beta}$ positive. Then the induced $u$ and $v$ from \eqref{uv} will be positive as expected. Also, we can write the collision map \eqref{YB1}, as 
\begin{equation} \label{colYB}
\mathcal{R}_{m_1,m_2}= (\phi^{-1} \times \phi^{-1}) \circ {R}_{m_1,m_2} \circ (\phi \times \phi), 
\end{equation} 
where $\phi:(-c,c) \rightarrow (0,+ \infty)$ is the bijection $\phi(v)=\sqrt{\frac{c+v}{c-v}}$ and 
$$(\phi \times \phi)(v_1,v_2):=(\phi(v_1),\phi(v_2)).$$ 
This transformation preserves the YB property, so  $\mathcal{R}_{m_1,m_2}$ is a YB map as well.
In addition, the strong Lax representation of $R_{\alpha,\beta}$ implies that the equation 
$$L(\phi(v_1'),m_1,\zeta)L(\phi(v_2'),m_2,\zeta)=L(\phi(v_2),m_2,\zeta)L(\phi(v_1),m_1,\zeta)$$  
is equivalent to $(\phi(v_1'),\phi(v_2'))=R_{m_1,m_2}(\phi(v_1),\phi(v_2))$, 
or 
$$(v_1',v_2')=(\phi^{-1} \times \phi^{-1}) \circ {R}_{m_1,m_2} \circ (\phi \times \phi)(v_1,v_2),$$
which shows that $L(\phi(v_1),m_1,\zeta)$ is a strong Lax matrix of  
$\mathcal{R}_{m_1,m_2}$. 
Finally, the invariant symplectic form $\omega$ of $R_{\alpha,\beta}$ implies the invariant symplectic form
$\tilde{\omega}=(\phi \times \phi)^{*} \omega$ of the map 
$\mathcal{R}_{m_1,m_2}$. 
We summarize our results in the following theorem.

\begin{theorem} \label{thmcolYB}
The collision map \eqref{YB1}, $\mathcal{R}_{m_1,m_2}:\mathcal{I} \times \mathcal{I} \rightarrow \mathcal{I} \times \mathcal{I}$, 
where $\mathcal{I}=(-c,c)$, coincides with the map \eqref{colYB} and it is a symplectic parametric YB map with strong Lax matrix 
$$\mathcal{L}(v_1,m_1,\zeta)=L(\phi(v_1),m_1,\zeta).$$
\end{theorem}

The YB property of the collision map reflects the fact that the resulting velocities of the collision of three particles are independent 
of the ordering of the collisions. This can be generalized for more particle collisions taking into account the commutativity of the transfer maps as defined in \cite{ves2}.

Here, we restricted our analysis in the case where the two masses $m_1,m_2$ of the particles are conserved after the collision. 
In the most general situation of relativistic elastic collisions, one can consider to have four different masses instead of two. 
In this case, the conservation of relativistic energy and momentum leads to a correspondence rather than a map.The integrability aspects of this correspondence will be left for future investigation. 

\begin{remark} \label{remclas}

The quadrirational YB map \eqref{YBab},  $R_{\alpha,\beta}:(x,y)\mapsto (u,v)$, corresponds to 
the $H_{III}^{A}$ map of the classification list in \cite{papclas}  
under the transformation $x \mapsto \alpha \hat{x}$, $y \mapsto \beta \hat{y}$,  $u \mapsto \alpha \hat{u}$, $v \mapsto \beta \hat{v}$
and the reparametrization $\alpha^2=\hat{\alpha}, \ \beta^2=\hat{\beta}$, 
while the transformation 
$x \mapsto -\alpha x$, $y \mapsto \beta \hat{y}$, $u \mapsto \alpha \hat{u}$, $v \mapsto -\beta \hat{v}$ 
and the same reparametrization leads to the $F_{III}$ map of the classification list in  
\cite{ABS2}. 
On the other hand, YB \eqref{YBab} can be reduced   
from a (non-involutive) 4-parametric YB map 
$$\hat{R}_{(\alpha_1,\alpha_2),(\beta_1,\beta_2)}(x,y)=(y \ P,x \  P), \ \text{with} \ P=\frac{\beta_1 x+\alpha_2 y}{\alpha_1 x+\beta_2 y},$$ that was presented 
in \cite{kp4}, by considering  
$R_{\alpha,\beta}:={\hat{R}}_{(\beta,\beta),(\alpha,\alpha)}$. 
As it was shown in  \cite{kp4} following  \cite{shib},
the map $\hat{R}_{(\alpha_1,\alpha_2),(\beta_1,\beta_2)}$ 
is associated with the (non $D_4$-symmetric) 3D consistent quad-graph equation 
$$w_{k+1,l}(\beta_1 w_{k,l}+\alpha_2 w_{k+1,l+1})-w_{k,l+1}(\alpha_1 w_{k,l}+\beta_2 w_{k+1,l+1})=0.$$
By setting $\alpha_1=\alpha_2=\beta$ and $\beta_1=\beta_2=\alpha$ we derive the 3D consistent equation associated with the 
YB map \eqref{YBab}
$$w_{k+1,l}(\alpha w_{k,l}+\beta w_{k+1,l+1})-w_{k,l+1}(\beta w_{k,l}+\alpha w_{k+1,l+1})=0.$$
This is a discrete version of the potential modified KdV equation originally presented in \cite{NQC} and corresponds under 
a gauge tranformation to $H_3$ equation in \cite{ABS1} for $\delta=0$. The latter equation follows directly from the YB map \eqref{YBab}
by setting $x=w_{k+1,l}w_{k,l}$, $y=w_{k+1,l+1}w_{k+1,l}$, $u=w_{k+1,l+1}w_{k,l+1}$ and $v=w_{k,l+1}w_{k,l}$. 
 \end{remark}

\section{Transfer maps as periodic sequences of collisions and integrability}
For any YB map, we can define different families of multidimensional transfer maps and corresponding monodromy matrices from their Lax representation \cite{kp2,ves2,ves3}. 
Under some additional conditions the transfer maps turn out to be integrable. 

Here we will consider one family of transfer maps 
derived from the so-called {\it standard periodic staircase initial value problem} that was originally presented as a periodic initial-value problem for  integrable partial difference equations in \cite{PNC,QCPN}.  In this framework (see e.g. \cite{kp2}), we define the {\it transfer map} as the $2n$-dimensional map 
$$T_n:(x_1,x_2, \dots, x_n,y_1,y_2 \dots, y_n) \mapsto (x_1',x_2' \dots, x_n',y_2', y_3'\dots, y_n',y_1'),$$
where $(x_i',y_i')=R_{\alpha,\beta}(x_i,y_i)$, and the {\it $k$-transfer map} as the map 
$T_n^k:=\underbrace{T_n\circ \dots \circ T_n}_k$. 
In this case, the {\it monodromy matrix} is defined as 
$$M_n(x_1,\dots,x_n,y_1, \dots, y_n)= \prod_{i=0}^{n-1}L(y_{n-i},\beta,\zeta)L(x_{n-i},\alpha,\zeta),$$
where $L$ is the Lax matrix of the YB map $R_{\alpha,\beta}$ (the elements in the product are arranged from left to right). 
Now, we observe that $$L(y_1',\beta,\zeta) M_n(x_1,\dots,x_n,y_1, \dots, y_n)=M_n(T_n(x_1,\dots,x_n,y_1, \dots, y_n)) L(y_1',\beta,\zeta).$$
So, it follows directly that the transfer map preserves the spectrum of the corresponding monodromy matrix. 

In a more general setting we can consider the non-autonomous case with different parameters $\alpha_i, \beta_i$ and 
$(x_i',y_j')=R_{\alpha_i,\beta_j}(x_i,y_j)$.  
Then the corresponding
$n$-transfer map $T_n^n:(x_1, \dots, x_n,y_1, \dots, y_n) \mapsto (x_1^{(n)}, \dots, x_n^{(n)},y_1^{(n)}, \dots, y_n^{(n)})$ 
preserves the spectrum of the monodromy matrix 
$M_n(x_1,\dots,x_n,y_1, \dots, y_n)= \prod_{i=0}^{n-1}L(y_{n-i},\beta_{n-i},\zeta)L(x_{n-i},\alpha_{n-i},\zeta).$ 
However, next we will restrict in the autonomous case where $\alpha_i=\alpha$ and $\beta_i=\beta$.

Let us now consider the YB map \eqref{YBab}. Even though the map itself is an involution, the corresponding transfer maps are not trivial. 
We will show that for any $n$ there are always three integrals of the transfer map $T_n$ and an invariant volume form. 

\begin{proposition} \label{propInt}
Any transfer map $T_n$ of the parametric YB map \eqref{YBab} has three first integrals 
\begin{eqnarray*}
E_n &=&\sum_{i=1}^n ( \alpha( x_i+\frac{1}{x_i})+\beta (y_i+\frac{1}{y_i}) ), \\  
P_n &=&\sum_{i=1}^n ( \alpha( x_i-\frac{1}{x_i})+\beta (y_i-\frac{1}{y_i}) ), \\
H_n &=& \frac{x_1 \dots x_n}{y_1 \dots y_n}+\frac{y_1\dots y_n}{x_1 \dots x_n} 
\end{eqnarray*}
and preserves the 
$2n$-volume form 
$$\Omega_n=(\frac{1}{x_1^2 x_2^2 \dots x_n^2}-\frac{1}{y_1^2 y_2^2\dots y_n^2})
dx_1 \wedge dx_2  \dots  \wedge dx_n  \wedge dy_1  \wedge dy_2  \dots \wedge dy_n.$$
\end{proposition} 

\begin{proof}
Let us denote $(X_1,X_2, \dots X_n, Y_1,Y_2, \dots Y_n):=T_n(x_1, \dots, x_n,y_1, \dots, y_n)$. That 
means that $X_i=x_i'$, for $i=1,\dots, n$, $Y_i=y'_{i+1}$, for $i=1,\dots, n-1$ and $Y_n=y '_1$,  
where $(x_i',y_i'):=R_{\alpha,\beta}(x_i,y_i)$. 

$H_n$ is derived from the invariance condition 
\begin{equation} \label{prinv}
\frac{x_i'}{y_i'}=\frac{y_i}{x_i}
\end{equation}
 of the YB map  \eqref{YBab}. 
From this invariant condition we derive  
$$\frac{X_1 \dots X_n}{Y_1 \dots Y_n}+\frac{Y_1\dots Y_n}{X_1 \dots X_n}=
\frac{x_1' \dots y_n'}{y_1' \dots y_n'}+\frac{y_1'\dots y_n'}{x_1' \dots x_n'}=
\frac{y_1\dots y_n}{x_1 \dots x_n}+\frac{x_1 \dots x_n}{y_1 \dots y_n}.$$ 
In a similar way we can show that $E_n$ and $P_n$, which correspond to the preservation of relativistic energy and momentum
\eqref{conen2}-\eqref{conmom2}, are also preserved by the map $T_n$.  

Furthermore, we have 
$$dx_i'\wedge dy_i'= -\left( \frac{\alpha x_i +\beta y_i}{\beta x_i+\alpha y_i} \right) ^2 dx_i \wedge dy_i=
- \frac{x_i'  y_i'}{x_i y_i}  dx_i \wedge dy_i$$
and 
$$\frac{1}{x_1' \dots  x_n' y_1' \dots y_n'}dx_1' \wedge dy_1' \dots \wedge dx_n'  \wedge dy_n'
=\frac{(-1)^n}{x_1 \dots  x_n y_1 \dots y_n}dx_1 \wedge dy_1  \dots \wedge dx_n  \wedge dy_n.$$
For any $n$, even or odd, we can rearrange the wedge product to derive 
\begin{equation} \label{prvVol}
\frac{1}{X_1 \dots  X_n Y_1 \dots Y_n}dX_1 \wedge  \dots  dX_n  \wedge dY_1 \wedge \dots  dY_n=
\frac{-1}{x_1 \dots  x_n y_1 \dots y_n}dx_1 \wedge \dots   dx_n  \wedge dy_1 \wedge \dots  dy_n.
\end{equation}
Now, we consider the function $$h_n(x_1,\dots, x_n,y_1, \dots y_n)=\frac{x_1 \dots x_n}{y_1 \dots y_n}-\frac{y_1\dots y_n}{x_1 \dots x_n}.$$
The invariant condition \eqref{prinv} implies that 
$$h_n(X_1,\dots, X_n,Y_1, \dots Y_n)=-h_n(x_1,\dots, x_n,y_1, \dots y_n).$$ 
Therefore, from \eqref{prvVol} we conclude that the $2n$-form 
$$-\frac{h_n(x_1,\dots, x_n,y_1, \dots y_n)}{x_1 \dots  x_n y_1 \dots y_n}dx_1 \wedge \dots  \wedge dx_n  \wedge dy_1 \wedge \dots \wedge dy_n,
$$ which coincides with $\Omega_n$, 
is invariant under $T_n$.
 
\end{proof}

\begin{remark} We can combine the integrals $E_n$ and $P_n$ to derive the linear integral 
$$\frac{1}{2}(E_n+P_n)=\sum_{i=1}^n ( \alpha x_i+\beta y_i ).$$
Also, \eqref{prvVol} indicates that $T_n^2:=T_n \circ T_n$ preserves the $2n$-volume form 
$$\frac{1}{x_1 \dots  x_n y_1 \dots y_n}dx_1 \wedge \dots   dx_n  \wedge dy_1 \wedge \dots  dy_n.$$
\end{remark}

More integrals are obtained from the spectrum of the monodromy matrix $M_n$. 
The integrals of $T_n$ give rise to integrals of the transfer map $\mathcal{T}_n$ that corresponds to the collision 
YB map \eqref{colYB}. 

\subsection{Examples of six and four particle periodic collisions}
As an example, we will examine the  autonomous case for $n=3$. In this case, the transfer map $T_3$ is the six-dimensional map
\begin{eqnarray*}
&&T_3(x_1,x_2,x_3,y_1,y_2,y_3) \label{T3}
\\ &&=(y_1 \frac{\alpha x_1 + \beta y_1}{\beta x_1+\alpha y_1},y_2 \frac{\alpha x_2 + \beta y_2}{\beta x_2+\alpha y_2},
y_3 \frac{\alpha x_3 + \beta y_3}{\beta x_3+\alpha y_3}, 
x_2 \frac{\alpha x_2 + \beta y_2}{\beta x_2+\alpha y_2}, x_3 \frac{\alpha x_3 + \beta y_3}{\beta x_3+\alpha y_3}, x_1 \frac{\alpha x_1 + \beta y_1}{\beta x_1+\alpha y_1}).  \nonumber
\end{eqnarray*}
The corresponding monodromy matrix is 
$$M_3(\mathbf{x},\mathbf{y})=L(y_3,\beta,\zeta) L(x_3,\alpha,\zeta)L(y_2,\beta,\zeta) L(x_2,\alpha,\zeta)L(y_1,\beta,\zeta) L(x_1,\alpha,\zeta),$$
with $(\mathbf{x},\mathbf{y})=(x_1,x_2,x_3,y_1,y_2,y_3)$, and 
 $$TrM_3(\mathbf{x},\mathbf{y})=2 \zeta^6+I_2(\mathbf{x},\mathbf{y})\zeta^4+I_1(\mathbf{x},\mathbf{y})\zeta^2
 +I_0(\mathbf{x},\mathbf{y}).$$ 
 $I_0$, $I_1$ and $I_2$ are first integrals of the map $T_3$ and we can verify directly that they are functionally independent, 
 i.e. the Jacobian matrix of the integrals has full rank.   
Additionally, we have the two extra integrals $E_3$ and $P_3$ associated with the relativistic energy and momentum and the integral 
$H_3$ which in this case coincides with $I_0$. 
 
 All of the integrals together are functionally dependent (the Jacobian matrix of all the integrals has rank four) but we can choose four of them, like 
 $I_0$, $I_1$, $E_3$ and $P_3$ that are functionally independent. This generically leads to a 2-dimensional common level set where the orbit of 
 $T_3$ lies.  
A projection of an orbit of the map $T_3$ on $\mathbb{R}^3$  can be seen in figure \ref{torus}. 

\begin{figure}
\centering
\includegraphics[scale=0.29]{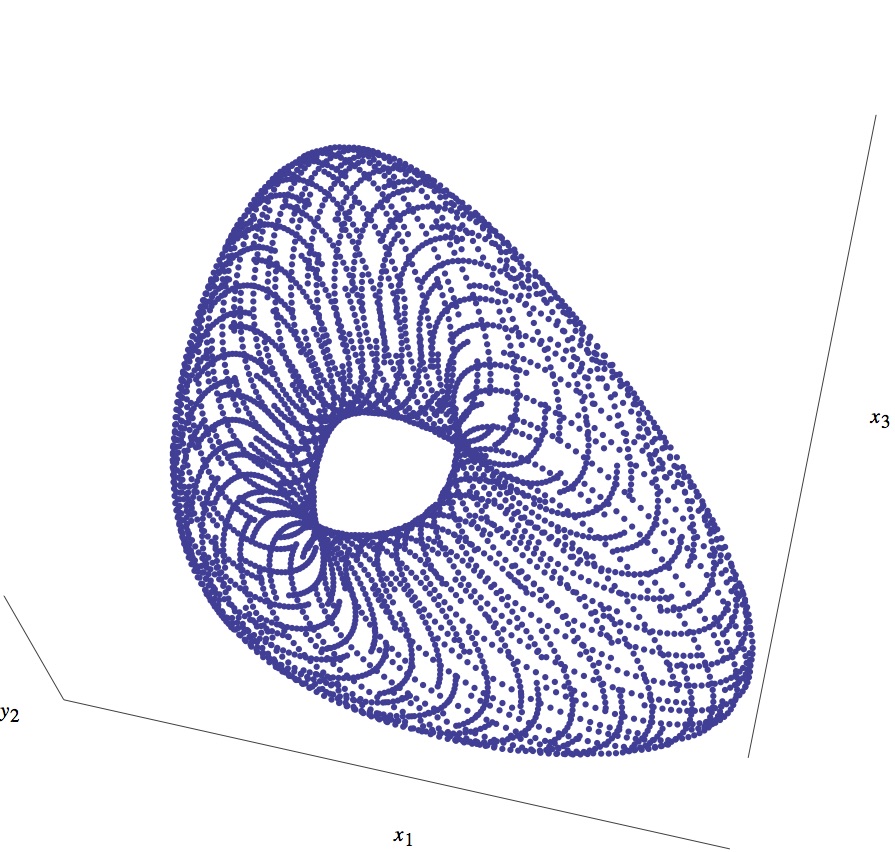}
\caption{Projection on $\mathbb{R}^3$ of an orbit of the map $T_3$ (for $\alpha=3$, $\beta=1$)}
 \label{torus}
\end{figure}

In addition, $T_3$ preserves the 6-volume form
$$\Omega_3=(\frac{1}{x_1^2 x_2^2 x_3^2}-\frac{1}{y_1^2 y_2^2 y_3^2})dx_1 \wedge dx_2  \wedge dx_3  \wedge dy_1  \wedge dy_2  \wedge dy_3.$$

The evolution of  $T_3$ represents a periodic sequence of collisions of six particles $p_1, \dots, p_6$.  Denoting  
by $p_i \leftrightarrow p_j$ the collision of the particles $p_i$ and $p_j$, initially we consider  the collisions 
$p_1 \leftrightarrow p_2$, $p_3 \leftrightarrow p_4$ and $p_5 \leftrightarrow p_6$. Subsequently, the collisions 
 $p_1  \leftrightarrow p_4$, $p_3  \leftrightarrow p_6$, $p_5  \leftrightarrow p_2$ and finally the collisions 
 $p_1 \leftrightarrow p_6$, $p_3  \leftrightarrow p_2$ and $p_5  \leftrightarrow p_4$. We proceed periodically with the same sequence of collisions. 
 Here we have considered that all the odd particles carry the same mass $\alpha$ and the even particles the mass $\beta$. 
 
 The case for $n=2$ is simpler.  According to proposition \ref{propInt}, the transfer map 
 $$T_2(x_1,x_2,y_1,y_2) =(y_1 \frac{\alpha x_1 + \beta y_1}{\beta x_1+\alpha y_1},y_2 \frac{\alpha x_2 + \beta y_2}{\beta x_2+\alpha y_2}, 
x_2 \frac{\alpha x_2 + \beta y_2}{\beta x_2+\alpha y_2}, x_1 \frac{\alpha x_1 + \beta y_1}{\beta x_1+\alpha y_1}),$$
has three functionally independent integrals, namely $E_2$, $P_2$ and $H_2$ 
and preserves the 4-volume form
$$\Omega_2=(\frac{1}{x_1^2 x_2^2 }-\frac{1}{y_1^2 y_2^2 })dx_1 \wedge dx_2   \wedge dy_1  \wedge dy_2 ,$$
therefore it is superintegrable\footnote{We call an n-dimensional map 
superintegrable if it admits $n-1$ functionally independent integrals and it is volume preserving.}.

The evolution of $T_2$ represents the periodic sequence of collisions of four particles 
$(p_1 \leftrightarrow p_2$, $p_3 \leftrightarrow p_4)$ and $(p_1 \leftrightarrow p_4, p_3 \leftrightarrow p_2)$. 

\subsection{Transfer maps as reductions of higher-dimensional integrable maps on an invariant manifold}
In a similar way we can treat any transfer map which represent different sequences of collisions and derive integrals from the spectrum of the corresponding monodromy matrix. 
Nevertheless, in order to claim complete integrability in the Liouville sense any transfer map in question should be symplectic with respect to 
a suitable symplectic structure 
and the integrals must be  in involution.  Even if we haven't found such symplectic structure, we will show that the transfer maps of the YB map \eqref{YBab} 
can be derived from the restriction of symplectic maps of twice dimension on an invariant manifold.  

We consider the four-dimensional YB map $$\tilde{R}_{\alpha,\beta}:((x_1,x_2),(y_1,y_2)) \mapsto ((u_1,u_2),(v_1,v_2))$$ that is derived from 
the unique solution with respect to $u_i$, $v_i$ of the Lax equation 
$$\tilde{L}(u_1,u_2,\alpha,\zeta)\tilde{L}(v_1,v_2,\alpha,\zeta)=\tilde{L}(y_1,y_2,\beta,\zeta)\tilde{L}(x_1,x_2,\alpha,\zeta),$$
where

\begin{equation} \label{Ltilde}
\tilde{L}(x_1,x_2,\alpha,\zeta)=\left(
\begin{array}{cc}
  \frac{\zeta}{\alpha}+x_1 & x_2 \\
 \frac{1-x_1^2}{x_2} &  \frac{\zeta}{\alpha}-x_1
\end{array}
\right).
\end{equation}
 The Lax matrix \eqref{Ltilde} has been constructed by restriction on particular symplectic leaves of the Sklyanin bracket on polynomial matrices 
\cite{kp1,kp3}.  The YB map $\tilde{R}_{\alpha,\beta}$ is a symplectic quadrirational YB map (case $I$ in \cite{kp3}). 
The corresponding invariant Poisson structure is given by the Sklyanin bracket \cite{skly1,skly2}, which in this case is equivalent to  
$$ \{x_1,x_2 \}=\frac{x_2}{\alpha}, \ \{y_1,y_2 \}=\frac{y_2}{\beta}, \ \{x_i,y_j \}=0.$$
This Poisson bracket can be extended to $\mathbb{R}^{2n} \times \mathbb{R}^{2n}$ as 
\begin{equation} \label{extPois}
 \sum_{i=1}^n ( \frac{x_{2i}}{\alpha} \frac{\partial} {\partial x_{2i-1}}\wedge \frac{\partial} {\partial x_{2i}} + 
\frac{y_{2i}}{\beta} \frac{\partial} {\partial y_{2i-1}}\wedge \frac{\partial} {\partial y_{2i}} ).
\end{equation}
Since this Poisson structure is obtained from the Sklyanin bracket, we can show that every transfer map 
$\tilde{T}_n$ of $\tilde{R}_{\alpha,\beta}$ will be Poisson with respect to
\eqref{extPois} and the integrals 
that are derived from the trace of the corresponding monodromy matrix will be in involution. 

Now, we observe that by setting $x_1=y_1=0$ at $\tilde{R}_{\alpha,\beta}$, we derive $u_1=v_1=0$. Particularly,   
the manifold 
$\mathcal{M}= \{ (0,x,0,y) / x,y \in \mathbb{R} \}$ is invariant under the map $\tilde{R}_{\alpha,\beta}$. 
 In addition, the reduced Lax matrix $\tilde{L}(0,x,\alpha,\zeta)$ is equivalent to the Lax matrix \eqref{Lax}.  
 Therefore we conclude that $\tilde{R}_{\alpha,\beta}$ is reduced to the map \eqref{YBab} on  the invariant manifold $\mathcal{M}$.  
 Consequently, any transfer map $\tilde{T}_n$ of $\tilde{R}_{\alpha,\beta}$ can be reduced to the transfer map 
 $T_n$ of the map \eqref{YBab} by setting $x_{2i-1}=y_{2i-1}=0$, for $i=1,\dots, n$.  
 
 In this sense,  for any $n$, $T_n$ can be regarded as a 
 reduction of the integrable map $\tilde{T}_n$ on the invariant manifold 
 $\mathcal{M}_n= \{ \left( (0,x_1),\dots, (0,x_n),(0,y_1), \dots, (0,y_n) \right) / x_i,y_i \in \mathbb{R}_{>0} \}$ and solutions of $\tilde{T}_n$ are reduced to solutions of 
 $T_n$, and consequently, to solutions of the corresponding transfer map of the collision map \eqref{YB1}. 
 


\section{Conclusions}
We proved that the change of velocities of two particles after elastic relativistic collision satisfy the YB equation. The corresponding map is  
equivalent to a quadrirational YB map that admits a Lax representation, which consequently implies a Lax representation of the original collision map. 
The $2n$-dimensional transfer maps, which represent particular sequences of $n$-particle periodic collisions, preserve the spectrum of the monodromy matrix and a volume form. The 4-dimensional transfer map turned out to be superintegrable, while the orbits of the 6-dimensional 
map lie on 
a two dimensional torus.  

Furthermore, we showed that the collision map can be regarded as a reduction of an integrable $4$-dimensional YB map on an invariant manifold and an equivalent reduction can 
be considered for any $2n$-dimensional transfer map. In this sense, we can consider the corresponding transfer maps as integrable.
However, strictly speaking we have not proved the Liouville integrability of the transfer maps due to the lack of a suitable Poisson structure 
and this is an issue that we would like to further investigate in the future. In the same framework it will be interesting to study the 
integrability of the transfer dynamics as defined by Veselov in \cite{ves2,ves3}, as well as reflection maps \cite{CZ} associated with 
relativistic collisions and fixed boundary initial value problems.

Regarding higher dimensional relativistic collisions, the conservation of relativistic energy and momentum is not enough to describe 
the resulting velocities and some additional assumptions have to be considered (for example the scattering angle of the particles).  In this way, 
 extra parameters will be involved in the resulting map. We aim to study in the future higher-dimensional collision maps and examine whether they  satisfy the YB equation.  

We conclude by remarking that this perspective of collisions as YB maps provides a remarkable link back to  
the origins of the YB equation 
in the field of statistical mechanics. 

\section*{Acknowledgement}
The author would like to thank Profs A.N.W. Hone, V.G. Papageorgiou and Dr. P. Xenitidis for the discussion and their useful comments. 
This research was supported by EPSRC (Grant EP/M004333/1).

\end{document}